\documentclass[11pt]{article} % use larger type; default would be 10pt

\usepackage[normalem]{ulem}

\usepackage{a4wide}
\usepackage{amsthm, amssymb, amsmath,hyperref,mathrsfs}
\usepackage{tikz-cd}

\newcounter{mnotecount}[section]

\renewcommand{\themnotecount}{\thesection.\arabic{mnotecount}}

\newcommand{\mnote}[1]%{}
{\protect{\stepcounter{mnotecount}}$^{\mbox{\footnotesize $%
\!\!\!\!\!\!\,\bullet$\themnotecount}}$ \marginpar{%\color{red}
\raggedright\tiny\em $\!\!\!\!\!\!\,\bullet$\themnotecount: #1} }

\newtheorem{Theorem}{Theorem}[section]
\newtheorem{Corollary}[Theorem]{Corollary}
\newtheorem{Lemma}[Theorem]{Lemma}

\newtheorem{Definition}[Theorem]{Definition}

\newtheorem{rem}[Theorem]{Remark}
\numberwithin{equation}{section}

\DeclareMathOperator{\Char}{Char}

\def\C{\mathbb{C}}

 \newcommand{\la}{\langle}
 \newcommand{\ra}{\rangle}
 \renewcommand{\l}{\left}
 \newcommand{\beq}{\begin{equation}}
 \newcommand{\eeq}{\end{equation}}
 \newcommand{\beqa}{\begin{eqnarray}}
 \newcommand{\eeqa}{\end{eqnarray}}

 \newcommand{\rf}[1]{(\ref{#1})}
 
 \newcommand{\D}[1]{{\cal D}(#1)}
 
 \newcommand{\M}{{\cal M}}

 \newtheorem{dfn}[Theorem]{Definition}
 \newtheorem{thm}[Theorem]{Theorem}
 
\newcommand{\norm}[2]{\|#1\|_{#2}}

\newcommand{\R}{\mathbb{R}}
\newcommand{\N}{\mathbb{N}}
\newcommand{\tx}{\tilde{x}}
\newcommand{\ty}{\tilde{y}}
\newcommand{\txx}{{\bf{x}}}
\newcommand{\txixi}{{\boldsymbol{\xi}}}
\newcommand{\txi}{\tilde{\xi}}
\newcommand{\teta}{\tilde{\eta}}
\newcommand{\tep}{\tilde{\epsilon}}

\newcommand{\ang}[1]{\langle#1\rangle}

\renewcommand{\r}{\right}
 \renewcommand{\l}{\left}

\DeclareSymbolFont{ASMa}{U}{msa}{m}{n}
\DeclareMathSymbol{\hrist}{\mathord}{ASMa}{"16}

\tikzset{
  symbol/.style={
    draw=none,
    every to/.append style={
      edge node={node [sloped, allow upside down, auto=false]{$#1$}}}
  }
}

\title{Adiabatic Ground States in Non-Smooth Spacetimes}

\author{Sanchez Sanchez, Yafet \and Schrohe, Elmar} 

\newcommand{\Addresses}{{% additional braces for segregating \footnotesize
  \bigskip
  \footnotesize

 Y. Sanchez Sanchez, \textsc{Leibniz University Hannover, Institute of Analysis, Welfengarten 1, 30167 Hannover, Germany}\par\nopagebreak
  \textit{E-mail address}: \texttt{yess@math.uni-hannover.de}

 \medskip
E. Schrohe, \textsc{Leibniz University Hannover, Institute of Analysis, Welfengarten 1, 30167 Hannover, Germany}\par\nopagebreak
 \textit{E-mail address}: \texttt{schrohe@math.uni-hannover.de}

}}

\begin{document}
\maketitle

%-------------------------------------------------------------------------------------------------------------------------------------------------------
% ABSTRACT SECTION
%-------------------------------------------------------------------------------------------------------------------------------------------------------
{
\begin{abstract}
Ground states are a well-known class of Hadamard states in smooth spacetimes.
In this paper we show that the ground state of the Klein-Gordon field in a non-smooth ultrastatic spacetime is an adiabatic state. 
The order of the state depends linearly on the regularity of the metric. We obtain the result by combining microlocal estimates for the causal propagator,  propagation of singularities results for non-smooth pseudodifferential operators  and eigenvalue asymptotics for elliptic operators of low regularity.   
\end{abstract}
}

%-------------------------------------------------------------------------------------------------------------------------------------------------------
% INTRO SECTION
%-------------------------------------------------------------------------------
------------------------------------------------------------------------

\section{Introduction}

The analysis of quantum fields in spacetimes where the metric is not smooth has 
two main motivations. First, there are several models of physical phenomena that 
require spacetime metrics with finite regularity. These include  models of gravitational collapse \cite{adler}, astrophysical objects \cite{stars}  and general relativistic fluids \cite{fluids}. Second, the well-posedness of Einstein's equations, viewed as a system of hyperbolic PDE requires spaces with finite regularity \cite{l2}.%  In general, spacetimes with finite regularity contain ``weak" singularities, in contrast with ``strong" singularities such as those found inside black holes and at the beginning 
%of the universe. 

In this paper we focus on the on scalar fields $\phi$ that satisfy the Klein-Gordon equation
\begin{equation} \label{wave2}
(\square_g+m^2)\phi:=g^{\mu\nu}\nabla_{\mu}\nabla_{\nu}\phi+m^{2}\phi=0
\end{equation}
\noindent
on a manifold $M=\mathbb{R}\times{\Sigma}$ where $\Sigma$ is a compact Cauchy hypersurface, $g^{\mu\nu}$ is the inverse metric tensor {of a ultrastatic metric}, $\nabla_{\mu}$ is the covariant derivative and $m^{2}$ is a positive real number.\\

{In the smooth setting, Fulling, Narcowich and Wald showed that the ground state in an ultrastatic spacetime is a Hadamard state \cite{fulling}. Later, Kay and Wald showed the (non)-existence of Hadamard states  in stationary spacetimes with a bifurcate Killing horizon \cite{kaywald}. Then, Radzikowski introduced a microlocal characterisation in terms of the wavefront set \cite{rad}. This result allowed for further constructions of these states, for example by Junker \cite{junker} and G{\'e}rard and Wrochna \cite{wrochna}.

%Later , Sahlmann  and Versch, generalised the result for vector-valued fields and passive state which include in particular ground states.  
{In a non-smooth spacetime} the quantisation requires in a first instance 
that the classical system be well-posed. Several results in this direction have been obtained for different degrees of regularity in the time and space variables \cite{colombini}. Moreover, even when one has classical well-posedness, the quantisation procedure is a significant 
further challenge. However, some progress has been made for certain degrees of spacetime
regularity. For example: Dereziński and Siemssen showed  the existence of {classical and nonclassical propagators} under weak regularity assumptions \cite{der,der2}. H{\"o}rmann, Spreitzer, Vickers and one of the authors gave the  construction of  quantisation functors that satisfy the Haag-Kastler axioms in the $C^{1,1}$ case \cite{green}. In this paper we prove that the ground state of the quantum linear scalar field is an adiabatic state and that the adiabatic order is a linear function with respect to the metric regularity (Theorem \ref{main}). }\\

{\bf{Outline of the paper:}} In Section 2, we show the algebraic quantisation of fields satisfying Eq.\eqref{wave2} in spacetimes of finite regularity. We give details about the construction of the algebra of observables and precise definitions of the states considered. In Section 3, we state the main definitions and theorems regarding non-smooth pseudodifferential operators. In Section 4, we focus on ultra-static spacetimes and show that the ground state is an adiabatic state. % In Section 5, we show that via the $GNS$ construction the ground state satisfies  Haag duality.   

\section{Quantum Field Theory in Non-smooth  Spacetimes}

The quantisation of the linear scalar field is a procedure to change the mathematical structure of the theory. On the one hand in the classical theory, the states  are represented by vectors in a symplectic space, $(V,\Xi)$, and the classical observables are defined as smooth functionals on $(V,\Xi)$. On the other hand, in the framework of algebraic quantisation, the quantum observables of the theory are represented as the elements of a unique up to $*$-isomorphism $C^{*}$-algebra which satisfies the canonical commutation relations (CCR) and the quantum states, $\omega$, are given by certain positive linear functionals on the $C^{*}$-algebra \cite{wald, bgp}. Below we give details of the quantisation procedure.

\subsection{Observables}

For a classical system with equations of motion given by Eq. (\ref{wave2}) in a globally hyperbolic spacetime $(M,g)$ of regularity $C^{1,1}$, it was shown that the space $(V,\Xi)$ is given by $V=H^1_{\text{comp}}(M)/{\text{ker G}}$ and $\Xi([f],[g])=([f],G[g])_{L_{\mathbb{R}}^{2}(M)}$ where $H_{\text{comp}}^1(M)$ denotes compactly supported function in the Sobolev space $H^1(M)$ and ${\text{ker G}}$ is the kernel of the causal propagator \cite{green}. In fact, this symplectic space is symplectically isomorphic to the classical phase space $(\Gamma, \sigma)$ given by the space $\Gamma:=H_{\text{comp}}^2(\Sigma) \oplus H_{\text{comp}}^1(\Sigma)$ of real-valued initial data with compact support and  the symplectic bilinear form $$\sigma(F_1,F_2)=\int_{\Sigma}[q_1p_2 - q_2p_1] dv$$  with $F_i:=(q_i, p_i) \in \Gamma, i = 1, 2$ and $dv$ the induced volume form on $\Sigma$.

Moreover, to the symplectic space $(V,\Xi)$ one can associate a $C^{*}$-algebra $\cal{A}$ that satisfies the CCR, known as the Weyl algebra. It is generated by the elements $W([f])$, $[f] \in V$, that satisfy 

$$
W([f])^{*}= W([f])^{-1}= W([-f]) 
$$
$$
W([f_1])W([f_2]) = e^{-\frac{i}{2}\Xi([f_1],[f_2])}W([f_1 + f_2]) 
$$

for all $[f], [f_1], [f_2] \in V $ (see e.g. \cite{bgp, green}).

As $(V,\Xi)$ and $(\Gamma, \sigma)$ are isomorphic as symplectic spaces,  one can construct a $C^{*}$-algebra, ${\cal{B}}$,{ using the map $\alpha:{\cal{A}}\mapsto {\cal{B}}$ given by  $\alpha\left(W([f])\right) := W((\rho^t_0 Gf, \rho^t_1 Gf))$ where  $\rho^t_o\phi:=\phi|_{\Sigma_t},\;\rho^t_1\phi:=\frac{\partial\phi}{\partial t}|_{\Sigma_t}$.} The algebra ${\cal{B}}$ is $*$-isomorphic to the Weyl algebra $\cal{A}$ described above. Each of these algebras represents the quantum observables of the theory.

Moreover, one can localise this construction to suitable subsets of $M$ following the approach of local quantum physics. In fact, one can do these local constructions in a functorial way and the functors satisfy the Haag-Kastler axioms ( see \cite[Theorem 6.12]{green}).

\subsection{States}

The quantum states as defined above need to be further restricted in order to be physically relevant. A candidate for physical quantum states,  $\omega$,  are quasifree states that satisfy the microlocal spectrum condition.

 To be precise, given a real scalar product $\mu:\Gamma\times\Gamma\rightarrow\mathbb{R}$ satisfying 
\begin {equation}\label{satu}
|\sigma(F_1, F_2)|^{2}\le \mu(F_1, F_1)\mu(F_2, F_2)
\end{equation}
 
for all $F_1, F_2 \in \Gamma$, there exist a quasifree state $\omega_{\mu}$ acting on the algebra $\cal{B}$ associated with $\mu$ given by $\omega_{\mu}(W(F))=e^{-\frac{1}{2}\mu(F,F)}$. Moreover, one can determine the (“symplectically smeared”) two-point function of $\omega_{\mu}$ by

\begin{equation}\label{mu}
\lambda(F_1, F_2) = \mu(F_1, F_2) + \frac{i}{2}\sigma(F_1, F_2)
\end{equation}

for $F_1, F_2 \in \Gamma$.  The Wightman two-point function $\omega^{(2)}_{\mu}$ associated to the state $\omega_{\mu}$,  is given by

\begin{equation}\label{twopoint}
\omega^{(2)}_{\mu}(f_1,f_2)=\lambda\left(\binom{\rho_0 Gf_1}{\rho_1 Gf_1},\binom{\rho_0 Gf_2}{\rho_1 Gf_2}\right)
\end{equation}

for $f_1,f_2\in H^1_\text {comp}(M)$. {By restricting the two point function $\omega^{(2)}_{\mu}$  to ${\D{M}}\otimes {\D{M}}$ one obtains a bidistribution in $M\times M$.}

To define the microlocal spectrum condition,  it is useful to introduce the sets
\vspace{-.2cm}
\begin{eqnarray}\label{C}
  &C=\big\{(\tx,\txi,\ty,\tilde{\eta})\in T^{*}(M\times M)\backslash0;  
  g^{ab}(\tx)\txi_{a}{\txi}_{b}=g^{ab}(\ty)\tilde{\eta}_{a}\tilde{\eta}_{b}=0, (\tx,\txi)\sim(\ty,\tilde{\eta})\big\} \\
  &C^{+}=\left\{(\tx,\txi,\ty,\tilde{\eta})\in C; \txi^{0}\ge0,\tilde{\eta}^{0}\ge0\right\}\nonumber,
  \end{eqnarray} 
\noindent
where $(\tx,\txi)\sim(\ty,\tilde{\eta})$ means that $\txi,\tilde{\eta}$ are cotangent to the null geodesic $\gamma$ at $\tx$ 
resp. $\ty$ and parallel transports of each other along $\gamma$.
 
 Using the above sets one can define the microlocal spectrum condition which goes back to Radzikowski\cite{rad}:
 \begin{Definition}
 A quasifree state $\omega_{H}$ on the algebra of observables satisfies the microlocal spectrum condition if its two point function $\omega^{(2)}_{H}$ is a distribution in $\mathcal{D}'(M\times M)$ and satisfies the following wavefront set condition
 
 $$WF'(\omega^{(2)}_{H})=C^{+},$$
 \noindent

where $WF'(\omega^{(2)}_{H}):= \{(x_1, \eta; x_2, -\tilde{\eta}) \in T^{*}(M\times M); (x_1, \eta; x_2, \tilde{\eta}) \in WF(\omega_{2H})\}.$
 
 \end{Definition}
 These states are called Hadamard states and include ground states in smooth spacetimes \cite{fulling, passive, junker, wrochna, fewster}.

 A larger class of states called adiabatic states of order $N$ characterised in terms of their Sobolev-wavefront set has been obtained by Junker and {one of the authors} \cite{adiabatic}. These states are the natural generalisation of Hadamard states suitable for spacetimes with limited regularity. 
 
 \begin{Definition}\label{adiadi}
 A quasifree state $\omega_{N}$ on the algebra  of observables
is called an adiabatic state of order $N\in\mathbb{R}$ if its two-point function $\omega^{(2)}_{N}$ is a
bidistribution that satisfies the following $H^s$-wavefront set condition for all $s\le N +\frac{3}{2}$
$$WF'^{s}(\omega^{(2)}_{N})\subset C^{+},$$
\noindent
 \end{Definition} 

where  $WF^{s}$ is a refinement of the notion of the wavefront set in terms of Sobolev spaces. To be precise, a distribution $u$ is {\em microlocally in} $H^{s}$ at $(x_0,\xi_0)\in T^*M\backslash 0$ if 
there exists a conic neighbourhood ${\Gamma}$ of $\xi_0$ and a smooth function $\varphi\in {\D{M}}$ with $\varphi (x_0)\neq 0$ 
such that $$\int_{{\Gamma}}\ang{\xi}^{2s}|{\cal{F}}(\varphi u)(\xi)|^{2}d^{n}\xi <\infty.$$
Otherwise we say that $(x_0,\xi_0)$ lies in the $s$-wave front set $WF^s(u)$.  

If $u$ is microlocally in $H^{s}$ in an open conic subset ${\Gamma}\subset T^*M\backslash 0$ we write $u\in H^s_{mcl}({\Gamma})$.

\section{Pseudodifferential Operators with Non-smooth  Symbols}

\subsection{Symbol Classes}

Let $\{\psi_j; j=0,1,\ldots\}$ be a Littlewood-Paley partition of unity on 
$\mathbb R^n$, i.e., a partition of unity $1=\sum_{j=0}^\infty\psi_j$, 
where $\psi_0\equiv 1$ for $|\xi|\le 1$ and $\psi_0\equiv 0$ for $|\xi|\ge2$ and 
$\psi_j(\xi) =  \psi_0(2^{-j}\xi)-\psi_0(2^{1-j}\xi)$. 
The support of $\psi_j$, $j\ge1$, then lies in an annulus around the origin of interior radius 
$2^{j-1}$ and exterior radius $2^{1+j}$.  
%$\supp(\psi_j)\sim \langle\xi\rangle\sim 2^j$ and $\psi_j(\xi)=\psi_1(2^{1-j}\xi)$ for $j\ge2$.

\begin{Definition}\label{Holder}
{\rm (a)} 
For $\tau\in (0,\infty)$, the H{\"o}lder space $C^\tau(\mathbb{R}^n)$ is the set of all functions $f$ with
\begin{equation}
\|f\|_{C^\tau}:=\displaystyle\sum_{|\alpha|\le [\tau]}\|\partial^\alpha_{x}f\|_{L^\infty(\mathbb{R}^n)}+\displaystyle\sum_{|\alpha|= [\tau]}\sup_{x\neq y}\frac{|\partial^\alpha_{x}f(x)-\partial^\alpha_{x}f(y)|}{|x-y|^{\tau-[\tau]}}<\infty.
\end{equation}

{\rm (b)} For $\tau\in \mathbb{R}$ the Zygmund space 
$C^\tau_{*}(\mathbb{R}^n)$ consists of all functions $f$ with 
\begin{equation}
\|f\|_{C^\tau_*}=\sup_j 2^{j\tau}\|\psi_j (D) f\|_{L^\infty}<\infty
\end{equation}
\end{Definition}
Here $\psi_j(D)$ is the Fourier multiplier with symbol $\psi_j$, i.e., 
$\psi_j(D)u = \mathcal F^{-1}\psi_j\mathcal F u$, 
where $(\mathcal Fu)(\xi) = (2\pi)^{-n/2} \int e^{-ix\xi} u(x)\, d^nx$ is the Fourier transform. 

We have the following relations $C^\tau=C^\tau_{*}$ if $\tau\notin \mathbb{Z}$ and $C^\tau\subset C^\tau_{*}$ if $\tau\in \mathbb{N}$.

We next introduce  symbol classes % of symbols $p(x,\xi)$ in $\mathbb{R}^n
of finite H\"older or Zygmund regularity{,} following Taylor % regularity in $x$. We follow 
\cite{taylor}. We use the notation $\langle\xi\rangle:=(1+|\xi|^2)^{\frac{1}{2}}$, $\xi\in \mathbb R^n$.

\begin{Definition}{\rm (a)} Let $0\le \delta <1$
A symbol $p(x,\xi)$ belongs to $C^\tau_* S^{m}_{1,\delta}$ if  $$\|D^{\alpha}_{\xi}p(x,\xi)\|_{C_*^{\tau}}\le C_\alpha\langle\xi\rangle^{m-|\alpha|+\tau\delta} \text{ and } |D^{\alpha}_{\xi}p(x,\xi)|\le C_\alpha\langle\xi\rangle^{m-|\alpha|}.$$

{\rm (b)}  We obtain the symbol class $C^\tau S^{m}_{1,\delta}$ for $\tau>0$ by requiring that
%replacing $C_*^\tau$ by $C^\tau$ in the above definition and requiring additionally that 
$$\|D^{\alpha}_{\xi}p(x,\xi)\|_{C^{j}}\le C_\alpha\langle\xi\rangle^{m-|\alpha|+j\delta}, \quad0\le j\le \tau, \text{ and } |D^{\alpha}_{\xi}p(x,\xi)|\le C_\alpha\langle\xi\rangle^{m-|\alpha|}.$$
%. \yafet{Is this equivalent to Taylors definition in 1.3.18 in non linear PDE?}

{\rm (c)} A symbol $p(x,\xi)$ is in  $C^{\tau}S_{cl}^{m}$ provided $p(x,\xi)\in C^{\tau}S^{m}_{1,0}$ and $p(x,\xi)$ has a classical expansion 
$$p(x,\xi)\sim \sum_{j\ge0}p_{m-j}(x,\xi)$$ 
in terms $p_{m-j}$ homogeneous of degree $m-j$ in $\xi$ for $|\xi|\ge 1$, in the sense that the difference between $p(x,\xi)$ and the sum over $0\le j< N$ belongs to $C^{\tau}S^{m-N}_{1,0}$.

\end{Definition}

{

\subsection{Characteristic Set and Pseudodifferential Operators}
Let $p\in C^\tau S^{m}_{\rho,\delta}$, $\tau>0$,  with $\delta<\rho$. Suppose that there is a conic neighborhood $\Gamma$ of $(x_0,\xi_0)$ and constants $c,C>0$ such that 
$|p(x,\xi)|\ge c|\xi|^m$ for $(x,\xi) \in \Gamma$, $|\xi|\ge C$. Then $(x_0,\xi_0)$ is called {\em non-characteristic}.  If $p$ has a principal homogeneous symbol $p_m$, the condition is equivalent to $p_m(x_0,\xi_0)\neq 0$. The complement of the set of non-characteristic points  is  the set of characteristic points denoted by ${\Char}(p)$.

\begin{rem}
The Klein-Gordon operator on $M$ is given by   
\begin{equation}\label{kg}
P\phi=\partial_{tt}\phi-\Delta_h\phi +m^2\phi
\end{equation}
It has the symbol $P(\tilde{x},\tilde{\xi})=(-\xi_0^2+h^{ij}\xi_i\xi_j)+i\frac{1}{\sqrt{h}}\partial_{x^i}(h^{ij}\sqrt{h})\xi_j+m^2$. 
For a metric of regularity $C^{\tau}$, the symbol $P(\tilde{x},\tilde{\xi})$ belongs to $C^{\tau-1}S^{2}_{cl}$ and $${\Char}(P)=\{(t,x,\xi_0,\xi)\in T^*M; -\xi_0^2+h^{ij}\xi_i\xi_j=0\}.$$
\end{rem}
}

The pseudodifferential operator $p(x,D_x)$ with the symbol $p(x,\xi)\in C^\tau S^m_{1,\delta} $ is given by 
\begin{equation}
p(x,D_x)u=(2\pi)^{-n/2} \int_{\mathbb{R}^n}e^{ix\cdot\xi}p(x,\xi)({\cal{F}}{u})(\xi)d^n\xi,
\quad u\in \mathcal S(\mathbb R^n). 
\end{equation}
It extends to continuous maps 
\begin{equation}
p(x,D_x): H^{s+m}(\mathbb{R}^n)\rightarrow H^{s}(\mathbb{R}^n) , \quad -\tau(1-\delta)<s<\tau.
\end{equation}

%\begin{Definition}

%\end{Definition}

%\begin{rem}
%A symbol $a(x,\xi)\in C^\tau S^{-\infty}_{1,\delta}(\Omega)$ if and only if $a(x,\xi)\in C^\tau S^{-N}_{1,\delta}$ for all $(x,\xi)\in\Omega, N\in \mathbb{N}$. Notice that in this case 
%$a(x,D_x): D'(\mathbb{R}^n)\rightarrow H^{s}(\mathbb{R}^n)$ for  $-\tau(1-\delta)<s<\tau$.
%\end{rem}

%\begin{Definition}
%
%\end{Definition}

%\begin{Definition}
%We define the essential support of $a(x,D)$ ( and also $a(x,\xi)$ to be the smallest closed conic subset of  $T^*(\mathbb{R}^n)\backslash \{0\}$ on the complement of which $a(x,\xi)$ has order $-\infty$. we denote this set by $ES(a)$.
%\end{Definition}

%The symbol smoothing which decomposes a non-smooth symbol into a smooth symbol with a non-smooth lower order term is done as follows.

\subsection{Symbol Smoothing} Given $p(x,\xi)\in C^{\tau}S_{1,\delta}^{m}$ and $\gamma\in (\delta,1)$ let 
\begin{equation}
p^{\#}(x,\xi)=\displaystyle\sum_{j=0}^{\infty}J_{\epsilon_{j}}p(x,\xi)\psi_j(\xi).
\end{equation}
Here $J_\epsilon$ is the smoothing operator given by 
$(J_\epsilon f)(x)=(\phi(\epsilon D)f)(x)$ with $\phi\in C^{\infty}_0(\mathbb{R}^n)$, $\phi(\xi)=1$ for $|\xi|\le 1$, and we take $\epsilon_j=2^{-j\gamma}$.

Letting $p^{b}(x,\xi)=p(x,\xi)-p^{\#}(x,\xi)$ we obtain the decomposition
\begin{equation}\label{smoothing}
p(x,\xi)=p^{\#}(x,\xi)+p^{b}(x,\xi),
\end{equation}
where $p^{\#}(x,\xi)\in S^{m}_{1,\gamma}$ and $p^{b}(x,\xi)\in C^{\tau}S^{m-\tau(\gamma-\delta)}_{1,\gamma}$.

If $p\in C^\tau S^m_{1,0}$, then we additionally have $p^b\in C^{\tau-t}S^{m-t\delta}_{1,0} $ with $\tau-t>0$ by \cite[Proposition 1.3.B]{taylor}.  Furthermore,  we have  better estimates, see \cite[Proposition 1.3.D]{taylor}: 
\begin{align}\label{morederivatives}
D_x^\beta p^{\#}(x,\xi)\in \begin{cases}S^{m}_{1,\delta},\quad &|\beta|\le \tau \text{ and}\\
%D_x^\beta p^{\#}(x,\xi)\in 
S^{m+\delta(|\beta|-r)}_{1,\delta}, \quad & |\beta|>\tau\end{cases} 
\end{align}

%One also has the following product theorems

%\begin{Theorem}
%Let $p\in C^\tau S^{m_1}_{1,\delta_1}$ and $q\in S^{m_2}_{1,\delta_2}$, $m_1,m_2\in \mathbb{R}$, $\tau>0$ and $0\le \delta_1,\delta_2<1$. We denote by $p\#q$ the symbol of $p(x,D)q(x,D)$ which has an asymptotic expansion $$p\#q\sim \displaystyle \sum_{\alpha\in \mathbb{N}^n_0}\frac{1}{\alpha!}D^\alpha_\xi p\partial_x^\alpha q.$$
%\end{Theorem}

\section{Ground States in Ultrastatic Spacetimes}

 Let $M=\R\times \Sigma$ where $\Sigma$ is a $3$-dimensional compact manifold and the Lorentzian metric $g$ is of the form  
$$ds^2=dt^2-h_{ij}(x)dx^idx^j$$ 
where $h_{ij}(x)$ are the components of a time independent   Riemannian metric
of H\"older regularity $C^{\tau}$ (when $\tau \in \N$ we will  consider the Zygmund spaces $C_*^\tau$, introduced in Definition \ref{Holder}).

Moreover, the vector field $\partial_t$ induces a one-parameter group of
isometries $\tau_t:M\to M,t\in{\R}$, such that
$\tau_t(\Sigma_{t_o})=\Sigma_{t_o+t}$. This group induces a one-parameter group of automorphisms { in the $C^*$-algebras }  as follows. Define ${\cal T}(t):\Gamma\to
\Gamma$ by
\[{\cal T}(t)\tilde{F}_{t_o}:=\tilde{F}_{t_o+t},
\]
where $\tilde{F}_{s}:=(\rho^s_o\phi,\rho^s_1\phi)$ and $\phi\in \text{ ker } (\square_g+m^2)$.

Since the symplectic form $\sigma$ is invariant under the action of
${\cal T}(t)$ and since ${\cal T}(t){\cal T}(s)={\cal T}(t+s)\; t,s\in {\mathbb{R}},\;{\cal T}$ is a one-parameter group
of symplectic transformations (also called Bogoliubov
transformations). It gives rise to a group of automorphisms $\tilde{\alpha} (t),
t\in{\R}$, (Bogoliubov automorphisms) on the algebra ${\cal B}$ via
\[\tilde{\alpha}(t) W(F)=W({\cal T}(t) F).
\]
In this case, there exists a preferred class of states on ${\cal A}$,
namely those invariant under $\alpha(t)$. A quasifree state $\omega_\mu$
will be invariant under this symmetry if and only if
\[\mu({\cal T}(t)F_1,{\cal T}(t)F_2)=\mu(F_1,F_2)\; \forall t\in{\R}
\;\forall F_1,F_2\in\Gamma.
\]
{

The specification of $\mu$ is equivalent to the specification of a one-particle structure as established by the following theorem of Kay and Wald \cite[Proposition 3.1]{kaywald}:

\begin{thm}\label{one}
Let $\omega_\mu$ be a quasifree state on ${\cal{B}}[\Gamma,\sigma]$.
Then there exists a {\rm\bf one-particle Hilbert space structure},
i.e.~a Hilbert space ${\cal H}$ and a real-linear map $k:\Gamma\to
{\cal H}$ such that \\
i) $k\Gamma +ik\Gamma$ is dense in ${\cal H}$,\\
ii) $\mu(F_1,F_2)={\rm Re}\la kF_1,kF_2\ra_{\cal H}\;\forall
F_1,F_2\in \Gamma$,\\
iii) $\sigma(F_1,F_2)=2{\rm Im}\la kF_1,kF_2\ra_{\cal H}\;
\forall F_1,F_2\in\Gamma$.\\
The pair $(k,{\cal H})$ is uniquely determined up to
unitary equivalence.
Moreover, $\omega_\mu$ is pure if and only if $k(\Gamma)$ is dense. 
\end{thm}

\begin{rem}
Notice that the specification of a Hilbert space $\cal{H}$ together with a real-linear map $k:\Gamma\to
{\cal H}$  such that  $k\Gamma +ik\Gamma$ is dense in ${\cal H}$ and $2{\rm Im}\la kF_1,kF_2\ra_{\cal H}=\sigma(F_1,F_2)$ gives rise  via Eq.\eqref{mu} to a real scalar product $\mu$ satisfying Eq.\eqref{satu}.
\end{rem}
}

Moreover, the automorphism group $\tilde{\alpha}(t)$ can be unitarily implemented in the
one-particle Hilbert space structure $(k,{\cal H})$ of an invariant state
$\omega_\mu$, i.e.~there exists a unitary group $U(t), t\in{\R},$ on
${\cal H}$ satisfying
\beqa
U(t)k &=& k{\cal T}(t) \label{2.15}\\
U(t)U(s) &=& U(t+s). \nonumber
\eeqa

If $U(t)$ is strongly continuous it takes the form $U(t)=e^{-iht}$ for some
self-adjoint operator $h$ on ${\cal H}$.\\
We define now the notion of ground states following Kay \cite{ground}:
\begin{dfn}\label{dfn2.4}
Let the phase space $(\Gamma,\sigma,{\cal T}(t))$ be given.
A quasifree {\rm\bf ground state} is a quasifree state over ${\cal B}
[\Gamma,\sigma]$ with one-particle Hilbert space structure $(k,{\cal H})$
and a strongly continuous unitary group $U(t)=e^{-iht}$ (satisfying
\rf{2.15}) such that $h$ is a positive operator (the ``one-particle
Hamiltonian'').
\end{dfn}

In the ultrastatic case we define the ground state, $\omega_G$ by
the one-particle Hilbert space structure $(k_{G},{\cal H}_{G})$ 
\beqa\label{hg}
k_G:\Gamma &\to& {\cal H}_G:=L^2_{\bf C}(\Sigma_{t_0})\nonumber\\
F=(q,p)&\mapsto&\frac{1}{\sqrt{2}}\l(A^{1/4}q-iA^{-1/4}p\r) \label{3.10},
\eeqa
 where $A:=-\Delta_h\phi +m^2$ and  $t_0\in \R$ (invariance under time translations makes any choice of $t\in \R$ equivalent to any other) and the strongly continuous unitary group  is given by $U(t):=e^{iA^{\frac{1}{2}}t}$. 

The Wightman two-point function of $\omega_G$ is:\\

\beq
\omega^{(2)}_{G}(h_1,h_2) = \lambda_G\l({\rho^t_oGh_1\choose
\rho_1^tGh_1 },{\rho_o^tGh_2\choose \rho_1^tGh_2}\r),\label{3.11}
\eeq
for $h_1,h_2\in\D{\M}$. 

Moreover using Eq.~\rf{twopoint}, Eq.\eqref{hg} and Theorem \ref{one}
the ``symplectically
smeared two-point function'' $\lambda_G$ is given on the initial
data $F_i={q_i\choose p_i}\in \Gamma$ by Eq.\eqref{mu},
\beqa
\lambda_G(F_1,F_2)&=& \l\la k_G F_1,k_GF_2\r\ra_{L^2_{\C}(\Sigma)} \nonumber\\
&=&\frac{1}{2}\l\la A^{1/4}q_1-iA^{-1/4}p_1,A^{1/4}q_2-iA^{-1/4}p_2\r\ra
_{L^2_{\C}(\Sigma)}  \nonumber\\
&=&\frac{1}{2}\l\la(A^{1/2}q_1-ip_1),
A^{-1/2}\l(A^{1/2}q_2-ip_2\r)\r\ra
_{L^2_{\C}(\Sigma)}, \label{3.12}
\eeqa
since $A$ is selfadjoint. Combining \rf{3.11} and \rf{3.12} we obtain
\beq
\omega^{(2)}_{G}(h_1,h_2)=\frac{1}{2}\l\la\l(A^{1/2}\rho_o^t-i\rho_1^t\r)Gh_1,A^{-1/2}
\l(A^{1/2}\rho^t_o-i\rho^t_1\r)Gh_2\r\ra_{L^2_{\C}(\Sigma_t)}. \label{3.13}
\eeq
%or in a more transparent

The two-point function, $\omega^{(2)}_{G}$,  of the ground state,  $\omega_G$, is the Schwartz kernel of the operator  
$\displaystyle{{e^{i A^\frac{1}{2}(t-s)}}{A^{-\frac{1}{2}}}}$.
%This follows from the singular
 %integral representation (using $E'=-E$)
%\beq
%\omega^{(2)}_{G}(\tx,\ty)=-\frac{1}{2}\int_{\Sigma_{q}}d^3y\sqrt{h(\vec{y})}\,
%\overline{G(\tx;q,\vec{y})\l(A^{1/2}-i\stackrel{\leftarrow}{\d{t}}\r)}
%A^{-1/2}\l(A^{1/2}-i\stackrel{\rightarrow}{\d{t}}\r)G(q,\vec{y};\ty),
%\label{3.14}
%\eeq
%where $A$ acts on $\vec{y}\in\Sigma_q$ and the fact that $G(\tx,\ty)$ is the Schwartz Kernel of the operator $\displaystyle{\frac{\sin( A^\frac{1}{2}(t-s)}{A^{\frac{1}{2}}}}$.

Explicitly, for $u,v\in {\mathcal D}(M)$  we have 
\begin{align}
&\omega^{(2)}_{G}(u, v)=\int_{M}\left(\frac{e^{i A^\frac{1}{2}(t-s)}}{A^{\frac{1}{2}}}u\right)(s,y)v(s,y)dsdy,\nonumber
\end{align}
which gives the singular integral kernel representation
\begin{equation}\label{kernel}
\omega^{(2)}_{G}(t,x;s,y)=\sum_j \lambda_j^{-1} e^{i\lambda_j(t-s)}\phi_j(x)\phi_j(y).
\end{equation}

{where $\{\phi_j, j=1,2,\ldots\}$ is an orthonormal basis of eigenfunctions of $L^2(\Sigma)$  associated to the eigenvalues $\lambda_j^2$  of the operator $m^2I-\Delta_h$. } 

The proof that the ground state in an ultrastatic smooth globally hyperbolic space-time is a Hadamard state has been shown by different methods \cite{fulling, passive, wrochna, junker}.  In the following section we show that  the ground states is an adiabatic state in the non-smooth case.

{\subsection{Microlocal analysis for Bisolutions of the Klein-Gordon Operator}}

We write local coordinates on $\mathbb R\times \Sigma$ in the form  
\begin{eqnarray}\label{xcoord}%\lefteqn{}\\
\tx = (t,x), \ty=(s,y)
\end{eqnarray}
and the associated covariables as 
\begin{eqnarray}\label{xicoord}%\lefteqn{}\\
\txi = (\xi_0,\xi), \teta= (\eta_{0},\eta). % \txx= (\tx,\ty), \txixi=(\txi,\teta)
\end{eqnarray}
On the product $(\mathbb R\times \Sigma)\times( \mathbb R\times \Sigma)$ we use 
\begin{eqnarray}\label{doublecoord}%\lefteqn{}\\
\txx= (\tx,\ty), \txixi=(\txi,\teta).
\end{eqnarray}

%If $u$ is not  microlocally $H^{s}$ at $(p,\chi)$, then we say $(p,\chi)\in WF^s (u)$.

In the sequel we shall apply the Klein-Gordon operator also to functions and distributions on $M\times M$. Using the coordinates in Eqs.\eqref{xcoord},\eqref{xicoord} and \eqref{doublecoord}, we distinguish  the cases, where $P$ acts on the first set of variables $(t, x)$  or on the second {set} $(s, y)$, and write $P_{(t,x)}$ and $P_{(s,y)}$, respectively. 
The associated symbols $P_{(t,x)}(\txx,\txixi)$ and $P_{(s,y)}(\txx,\txixi)$  formally depend on the full set of (co-)variables $(\txx,\txixi)$, however,  only the (co-)variables associated with either $(t, x)$ or $( s,y)$ show up:
\begin{eqnarray}%\label{}\lefteqn{}\\
\nonumber
P_{(t,x)}(\txx,\txixi)&=&P_{(t,x)}(\tx,\txi,\ty,\teta)=\underbrace{({-\xi^{0}}^{2}+h^{ij}(x)\xi_i\xi_j)}_{p_2(\txx,\txixi)}+\underbrace{i\frac{1}{\sqrt{h}}\partial_{x^i}(h^{ij}\sqrt{h}(x))\xi_j}_{p_1(\txx,\txixi)}+\underbrace{m^2}_{p_0(\txx,\txixi)}.
\label{pr0}\\\nonumber
P_{(s,y)}(\txx,\txixi)&=&P_{(s,y)}(\tx,\txi,\ty,\teta)=(-{\eta_{0}}^{2}+h^{ij}(x)\eta_i\eta_j)+{i\frac{1}{\sqrt{h}}\partial_{y^i}(h^{ij}\sqrt{h}(y))\eta_j}+{m^2}.
\end{eqnarray}
In particular, 
\begin{eqnarray}%\label{}\lefteqn{}\\
\label{char3}
  \Char(P_{(t,x)})&=&\Char(P)\times T^{*}M\cup\{(\txx,\txixi)\in T^{*}(M\times M)\backslash\{0\}, \txi=0\}\\
\Char(P_{(s,y)})&=&T^{*}M\times \Char (P)\cup\{(\txx,\txixi)\in T^{*}(M\times M)\backslash\{0\}, \teta=0\}.\nonumber
\end{eqnarray}
%\end{equation}
%
%and $P_{(s,y)}(\txx,\txixi)$ given by
%
%\begin{equation}\label{pr0}
%P_{(s,y)}(\txx,\txixi)=:P_{(s,y)}(\tx,\txi,\ty,\teta)=\underbrace{(-{\eta^{0}}^{2}+h^{ij}(x)\eta_i\eta_j)}_{p_0(\txx,\txixi)}+\underbrace{i\frac{1}{\sqrt{h}}\partial_{y^i}(h^{ij}\sqrt{h}(y))\eta_j}_{p_1(\txx,\txixi)}+\underbrace{m^2}_{p_2(\txx,\txixi)}.
%\end{equation}
%\blue{ We state the main microlocal estimates required for the proof of Theorem \ref{main}.}

{Now we will state a microelliptic estimate tailored for bisolutions of the Klein-Gordon operator 

\begin{Theorem}\label{elliptic}Let the metric $g$ be of class $C^\tau$, $\tau>1$,
$0\le \sigma<\tau-1$ and $v\in H_{loc}^{2+\sigma-\tau+\epsilon}(M\times M)$ for some ${\epsilon}>0$ with 
 $P_{(t,x)}(\txx,D_\txx)v= P_{(s,y)}(\txx,D_\txx)v=0$. Then 
% , where the metric satisfies $g\in C^\tau$ for $\tau>2$ , $v\in H_{loc}^{2+\sigma-\tau\delta}(M\times M)$. Then for $1>\delta>0, -(1-\delta)(\tau-1)<\sigma<\tau-1$
$$WF^{\sigma+2}(v)\subset \Char(P_{(t,x)})\cap\Char(P_{(s,y)}).$$
\end{Theorem}

The proof can be found in \cite[Theorem 3.4]{causal}.

{
\begin{rem}\label{propagation1}
Applying the symbol smoothing directly to $P_{(t, x)}\in C^{\tau-1} S^2_{1,0}$ 
would leave us with $P_{(t, x)}^b\in C^{\tau-1}S^{2-{(\tau-1)}\delta}_{1,\delta}$.
Therefore, we smooth each of the non-smooth symbols ( the principal symbol and the sub-leading term) separately  to obtain the remainder $p_2^b+p_1^b$ for $p^b_2\in C^\tau S^{2-\tau\delta}_{1,\delta}$ and  $p^b_1\in C^{\tau-1} S^{1-(\tau-1)\delta}_{1,\delta}$. %Therefore,  given  $u\in H^{2+s-\tau\delta}$ we have $p^b_2u\in H^s, p^b_1u\in H^{s+1-\delta}$ and hence  
%$p_2^bu+p_1^bu\in H^s$ for $-{(1-\delta)}(\tau-1)<s<\tau-1$. 
\end{rem}

}
{Furthermore, the main results on the microlocal propagation of singularities  in the non-smooth setting  that we will apply can be found in 
\cite[Proposition 6.1.D]{taylor} or \cite[Proposition 11.4]{tools}.
In particular, the theorem below holds for spacetime metrics belonging to the space $C_*^2$ \cite[p.215]{tools}.}

\begin{Theorem}\label{propagation}
Let $u\in \mathcal D'(M\times M)$ solve $P_{(t,x)}u=f$. Let $\gamma$ be a integral curve of the Hamiltonian vector field $H_{p_{2}}$ with $p_2$ the principal symbol of $P_{(t,x)}$.
{ If for some $s\in \R$,}  $f\in H_{mcl}^s({\Gamma})$ and {$P_{(t,x)}^b u\in H_{mcl}^s({ \Gamma)}$} where $\gamma\subset{\Gamma}$ with ${\Gamma}$ a conical neighbourhood and $u\in H_{mcl}^{s+1}(\gamma(0))$ then $u\in H_{mcl}^{s+1}({\gamma})$. 
\end{Theorem}

\begin{rem}\label{taylor}
If $u\in H_{loc}^{2+s-\tau\delta}$, then $P_{(t,x)}^b u\in H_{loc}^s(M\times M)$, see Remark \eqref{propagation1}. 
%Moreover, using the divergence structure of the operator if $u\in H_{loc}^{1+s-\tau\delta}(M\times M), f\in H_{loc}^{s-1}(M\times M), u\in H_{mcl}^{s}({\gamma}(0))$, then 
%$u\in H_{mcl}^{s}({\gamma})$ for $-2(1-\delta)<s\le 2$; see {\rm\cite[p.210]{tools}} for details. 
\end{rem}

}

\subsection{The Microlocal Spectrum Condition}

Now we will show that the Wightman two-point function of the ground state described above satisfies Defintion \ref{adiadi}.
 We will assume throughout this section that the metric is of regularity $C^\tau$ with $\tau>2$.

Let $\{\phi_j\otimes\phi_k; j,k=1,2,\ldots\}$ be an orthonormal basis of $L^2(\Sigma)\otimes L^2(\Sigma) $ associated to the eigenfunctions $\{\phi_j\}$ {and the eigenvalues $\{\lambda_j^2\}$} of the operator $m^2I-\Delta_h$.  Then, for $u\in L^2(M\times M)$ we  have the representation  
\begin{eqnarray}\label{repu}%\lefteqn{}\\
{u(t,s, x,y) =\sum_{j,k}u_{jk}(t,s)\phi_j(x)\phi_k(y)}  \quad \text{with }
u_{jk}=\ang{u,\phi_j\otimes\phi_k}\in  L^2(\mathbb{R}^2).
\end{eqnarray}

Moreover, we have the following generalisation for $u\in H^{2\theta}(M\times M)$ shown in \cite[Proposition 4.1, Corollary 4.4]{causal}

\begin{Theorem}\label{normneg}
For $-1\le \theta\le 1$
\begin{align*}
&H^{2\theta}(\mathbb{R}^2\times \Sigma^2)
&=\{u\in \mathcal S'(\mathbb{R}^2\times \Sigma^2); \sum_{j,k}\int_{\mathbb{R}^2}(|\xi_0|^2+|\eta_0|^2+\lambda_j^2+\lambda_k^2)^{2\theta}|\mathcal Fu_{jk}(\xi_0, \eta_0)|^2d\xi_0d\eta_0<\infty\},
\end{align*}
with $u_{j,k} = \ang{u,\phi_j\otimes \phi_k}\in \mathcal S'(\R^2)$. 
\end{Theorem}

The previous theorem allows us to establish the local Sobolev regularity of the two-point function.

\begin{Theorem} $\omega^{(2)}_{G}\in H_{loc}^{-\frac{1}{2}-\tep}(M\times M)$.
\end{Theorem}
\begin{proof}

Let $\psi\in \mathcal D(M\times M)$. We will show $\psi \omega^{(2)}_{G}\in H^{-\frac{1}{2}-\epsilon}(M\times M)$.
%Without loss of generality we take $\psi(t,s,x,y)=\psi_1(t)\psi_2(s)$ with $\psi_1,\psi_2\in \mathcal D(\R) $. % and  $m^2=\frac{1}{2}$.
According to  Theorem \ref{normneg} 
\begin{align}
&\|\psi \omega^{(2)}_{G}\|^2_{H^{-\frac{1}{2}-\epsilon}(M\times M)}\\
&=\sum_{j=k}\int_{\mathbb{R}^2}\label{local}
(|\xi_0|^2+|\eta_0|^2+\lambda_j^2+\lambda_k^2)^{-\frac{1}{2}-\epsilon}
\Big|{\mathcal F_{(t,s)\to(\xi_0,\eta_0)}}
\Big(\frac{\psi(t,s)}{\lambda_j}e^{i{\lambda_j(t-s)}}\Big)(\xi_0, \eta_0)\Big|^2d\xi_0d\eta_0.
\end{align}

We have by direct computation that 
\begin{align}\nonumber
&\Big|{\mathcal F_{(t,s)\to(\xi_0,\eta_0)}}
\Big(\frac{\psi(t,s)}{\lambda_j}e^{i{\lambda_j(t-s)}}\Big)(\xi_0, \eta_0)\Big|^2\\
&=\frac{1}{\lambda_j^2}\left|{\mathcal F} (\psi)(\xi_0-\lambda_j, \eta_0+\lambda_j)\right|^2\nonumber\\
%&\le \frac{1}{\lambda_j^2}\left(\frac{C}{\ang{-\xi_0+\lambda_j}^N\ang{\eta_0+\lambda_j}^N}\right)^2\label{Fpsi}
\end{align}

Taking into account that $\norm{\psi}{L^2(\R^2)}=\norm{{\cal{F}}(\psi)}{L^2(\R^2)}<\infty$ we have (with constants possibly changing from line to line)

\begin{align*}
&\|\psi \omega^{(2)}_{G}\|_{H^{-\frac{1}{2}-\epsilon}(M\times M)}\\
&=\sum_{j=k}\int_{\mathbb{R}^2}
(|\xi_0|^2+|\eta_0|^2+\lambda_j^2+\lambda_k^2)^{-\frac{1}{2}-\epsilon}
\Big|{\mathcal F_{(t,s)\to(\xi_0,\eta_0)}}
\Big(\frac{\psi(t,s)}{\lambda_j}e^{i{\lambda_j(t-s)}}\Big)(\xi_0, \eta_0)\Big|^2
d\xi_0d\eta_0\\
&\le \displaystyle{\sum^\infty_{j=1}}\int_{\mathbb{R}^2}(\lambda_j^2+\lambda_j^2)^{-\frac{1}{2}-\epsilon}
\frac{1}{\lambda_j^2}\left|{\mathcal F}(\psi)(\xi_0-\lambda_j, \eta_0+\lambda_j)\right|^2d\xi_0d\eta_0\\
&\le C\displaystyle{\sum^\infty_{j=1}}\frac{(\lambda_j^2+\lambda_k^2)^{-\frac{1}{2}-\epsilon}}{\lambda_j^2}\\
&\le C \displaystyle{\sum^\infty_{j=1}}\frac{1}{\lambda_j^{3+2\epsilon}}.
\end{align*}

From Weyl's law for non-smooth metrics \cite[Theorem 1.1]{zielinski} we obtain the estimate $l^{\frac{2}{3}}\le C \lambda_j^2$ for a suitable constant $C$ which gives 
\begin{align*}
\|\psi \omega^{(2)}_ {G}\|^2_{H^{-\frac{1}{2}-\epsilon}(M\times M)}&\le \displaystyle{\sum^\infty_{j=1}}\frac{C}{\lambda_j^{3+2\epsilon}}\le \displaystyle{\sum^\infty_{j=1}}\frac{C'}{j^{1+\epsilon}}<\infty
\end{align*}
for a suitable constant $C'$.
\end{proof}

It will be useful to consider the  following bidistribution:
\begin{Corollary}
Let $\omega_A\in \mathcal D'(M\times M)$  be the bidistribution given by 
\begin{equation*}
\omega_A(u\otimes v):=-\int_{M\times M}\sum_j \lambda_l^{-2} e^{i\lambda_l(t-s)}\phi_l(x)\phi_l(y)\sqrt{h(y)}\sqrt{h(x)}  u(t,x) v(s,y)dsdydtdx
\end{equation*}
Then,
\begin{align}\label{normcor}
\omega_A\in H_{loc}^{\frac{1}{2}-\epsilon}(M\times M) \text{ for every } \epsilon >0.
\end{align}
\end{Corollary}

\begin{proof}
 
Direct computation shows that  for $\psi$ as in the previous proof 

\begin{align*}
&\|\psi \omega_A\|_{H^{s}(M\times M)}\\
&=\displaystyle{\sum_{j=1}^{\infty}}\int_{\mathbb{R}^2}(|\xi_0|^2+|\eta_0|^2+\lambda_j^2+\lambda_j^2)^{s}
\Big|{\mathcal F_{(t,s)\to(\xi_0,\eta_0)}}
\Big(\frac{\psi(t,s)}{\lambda_j^2}e^{-i{\lambda_j(t-s)}}\Big)(\xi_0, \eta_0)\Big|^2d\xi_0d\eta_0\\
&=\displaystyle{\sum_{j=1}^{\infty}}\int_{\mathbb{R}^2}(|\xi_0|^2+|\eta_0|^2+\lambda_j^2+\lambda_j^2)^{s}
\Big|{\mathcal F}
\Big(\frac{\psi}{\lambda_j^2}\Big)(\xi_0-\lambda_j, \eta_0+\lambda_j)\Big|^2d\xi_0d\eta_0\\
&=\displaystyle{\sum_{j=1}^{\infty}}\int_{\mathbb{R}^2}(|\xi_0-\lambda_j|^2+|\eta_0+\lambda_j|^2+\lambda_j^2+\lambda_j^2)^{s}
\Big|{\mathcal F}
\Big(\frac{\psi}{\lambda_j^2}\Big)(\xi_0, \eta_0)\Big|^2d\xi_0d\eta_0\\
&\le \displaystyle{\sum_{j=1}^{\infty}}\int_{\mathbb{R}^2}(2|\xi_0|^2+2|\eta_0|^2+3\lambda_j^2+3\lambda_j^2)^{s}
\Big|{\mathcal F}\Big(\frac{\psi}{\lambda_j^2}\Big)(\xi_0, \eta_0)\Big|^2d\xi_0d\eta_0\\
&\le C \displaystyle{\sum_{j=1}^{\infty}}\frac{(1+\lambda_j^2)^{s}}{\lambda_j^{4}}\int_{\mathbb{R}^2}(1+|\xi_0|^2+|\eta_0|^2)^s
\Big|{\mathcal F}\Big({\psi}\Big)(\xi_0, \eta_0)\Big|^2d\xi_0d\eta_0\\
&\le C\displaystyle{\sum_{j=1}^{\infty}}\frac{(1+\lambda_j^2)^{s}}{\lambda_j^{4}}\norm{\psi}{H^s(\R^2)}\\
&{\le \displaystyle{\sum_{j=1}^{j_0}}\frac{(1+\lambda_j^2)^{s}}{\lambda_j^{4}}+C\displaystyle{\sum_{j=j_0}^{\infty}}\frac{(\lambda_j^2)^{s}}{\lambda_j^{4}}}
\end{align*}
  where we have chosen $j_0$ large enough such that $\lambda_{j_0}>1$.% (that may be changing from line to line).

According to Weyl's law for non-smooth metrics \cite[Theorem 1.1]{zielinski} we have the estimate $l^{\frac{2}{3}}\le C \lambda_l^2$ for a suitable constant $C$. 
This gives  for $s=\frac{1}{2}-\epsilon$
\begin{align}
\|\psi \omega_A\|^2_{H^{\frac{1}{2}-\epsilon}(M\times M)}&\le C+\sum_l\frac{C}{\lambda_l^{3+2\epsilon}}\le C+ \sum_l\frac{C}{l^{1+\epsilon}}<\infty.
\end{align}
for a suitable constant $C$. % (that may be changing from line to line).
\end{proof}

\begin{rem}
Notice that $i\partial_{t}\omega_A=\omega^{(2)}_{G}$.
\end{rem}

\begin{Lemma}\label{char}
For any $\tilde{\epsilon}>0$ 
\begin{equation}
WF^{-\frac{1}{2}-\tilde{\epsilon}+\tau}(\omega^{(2)}_{G})\subset \Char(P)\times\Char(P).
\end{equation} 
\end{Lemma}

\begin{proof}
Since $\omega^{(2)}_{G}$ satisfies $(\partial_t+\partial_s)\omega^{(2)}_{G}=0$ we conclude that for all $l\in \mathbb{R}$
\begin{equation}\label{time}
WF^l(\omega^{(2)}_{G})\subset WF(\omega^{(2)}_{G})\subset \Char(\partial_t+\partial_s)=\{(\tx,\xi_0,\xi, \ty,\eta_0,\eta)\in T^*(M\times M)\backslash\{0\}; \xi_0+\eta_0=0\},
\end{equation}
where the second inclusion follows from the standard theory 
of  pseudodifferential operators.

Now we have $P_{(t,x)}(\txx,D_\txx)\omega_A=P_{(s,y)}(\txx,D_\txx)\omega_A=0$.
{Choose $\epsilon<\tilde\epsilon/2$. Since $\omega_A\in H_{loc} ^{\frac12-\frac{\epsilon}2}(M\times M) = 
H_{loc} ^{(\frac12- \epsilon)+\frac{\epsilon}2}(M\times M)$, 
an  application of  Theorem \ref{elliptic} with $\sigma=-\frac{3}{2}+\tau-{\epsilon}< \tau-1$ }
% \yafet{i have removed all the $\delta$s in this lemma.} 
shows that 
\begin{align*}
WF^{\frac{1}{2}+\tau-\tilde{\epsilon}}(\omega_A)\subset &\Char(P_{(t, x)})\cap \Char(P_{(s, y)});
\end{align*} 
{here we assume without loss of generality that $\epsilon$ is so small that 
$-\frac{3}{2}+\tau-{\epsilon}\ge0$}.
Eq.\eqref{char3} implies that 
\begin{align*}
WF^{\frac{1}{2}-\tilde{\epsilon}+\tau}(\omega_A)\subset &(\Char(P)\times \Char(P))\\
&\cup( \{(\txx,\txixi)\in T^{*}(M\times M)\backslash\{0\}; \txi=0,(\ty,\teta)\in  \Char(P)\}\\
&\cup  \{(\txx,\txixi)\in T^{*}(M\times M)\backslash\{0\}; (\tx,\txi)\in \Char(P), \teta=0\}.
\end{align*} 
If $\tilde \eta=0$, then $\eta_0=0$, and $\xi_0=0$ by   Eq. (\ref{time}). Since   $\Char P=\{(\tilde x,\tilde \xi);(\xi_0)^2={\sum_{i=1}^{3}}h^{ij}(x)\xi_i\xi_j\}$  we then have  $\tilde \xi=0$. Together with the corresponding argument for the case $\tilde \xi=0$ this shows that  
\begin{equation}\label{char2}
WF^{\frac{1}{2}-\tilde{\epsilon}+\tau}(\omega_A)\subset \Char(P)\times\Char(P);
\end{equation} 
 otherwise $0\in T^*(M\times M)$ will be in $WF^{\frac{1}{2}-\tilde{\epsilon}+\tau}(\omega_A)$.

Since $WF^{-\frac{1}{2}-\tilde{\epsilon}+\tau}(i\partial_{t}\omega_A)\subset WF^{\frac{1}{2}-\tilde{\epsilon}+\tau}(\omega_A)$  by \cite[Proposition B.3]{adiabatic},  we have 
\begin{equation*}
WF^{-\frac{1}{2}-\tilde{\epsilon}+\tau}(\omega^{(2)}_{G})=WF^{-\frac{1}{2}-\tilde{\epsilon}+\tau}(i\partial_{t}\omega_{A})\subset WF^{\frac{1}{2}-\tilde{\epsilon}+\tau}(\omega_A)\subset (\Char(P)\times\Char(P)).
\end{equation*}

\end{proof}

\begin{Theorem}\label{positive}
For all $s\in \R$, $WF^{s}(\omega^{(2)}_{G})\subset \{(\tx,\txi,\ty,\teta)\in T^*(M\times M); \txi^0> 0\}$ .
\end{Theorem}

\begin{proof}
We define $F:\R+i]0,\delta[\subset\mathbb{C}\rightarrow \mathcal D'(\Sigma\times M)$ for $\delta>0$  by

\begin{equation}
\ang{F(z),\psi_1(s)\psi_2(x)\psi_3(y)}=\int_{M}\sum_{j}e^{iz\lambda_j}e^{-is\lambda_j}\left(\int_{\Sigma}\psi_2(x)\phi_j(x)dx\right)\phi_{j}(y)\psi_1(s)\psi_3(y)dyds.
\end {equation}

Notice that $\partial_zF:\R+i]0,\delta[\subset\mathbb{C}\rightarrow\mathcal D'(\Sigma\times M)$ is  given by

\begin{equation}
\ang{\partial_zF(z),\psi_1(s)\psi_2(x)\psi_3(y)}=i\int_{M}\sum_{j}\lambda_je^{iz\lambda_j}e^{-is\lambda_j}\left(\int_{\Sigma}\psi_2(x)\phi_j(x)dx\right)\phi_{j}(y)\psi_1(s)\psi_3(y)dyds
\end {equation}

and therefore $F$ is a holomorphic function with values in  $\mathcal D'(\Sigma\times M)$\cite[Theorem 10.11]{kaballo}. Moreover,  for $\varphi(t)\in \mathcal D(\R)$ we have 

\begin{align}
&|\ang{\ang{F(t+i\epsilon),\psi_1(s)\psi_2(x)\psi_3(y)},\varphi(t)}|\\
=&\left|\int_{\R}\int_{M}\sum_{j}e^{i(t+i\epsilon)\lambda_j}e^{-is\lambda_j}\left(\int_{\Sigma}\psi_2(x)\phi_j(x)dx\right)\phi_{j}(y)\psi_1(s)\psi_3(y)dyds\varphi(t)dt\right|\\
\le&\left|\int_\R\int_\R\sum_{j}\underbrace{e^{i(t+i\epsilon)\lambda_j}e^{-is\lambda_j}(\psi_2,\phi_j)_{L^2(\Sigma)}(\psi_3,\phi_j)_{L^2(\Sigma)}\psi_1(s)\varphi(t)}_{h_\epsilon(s,t,j)}dsdt\right|
\end{align}

Now let $g(s,t,j):=|(\psi_2,\phi_j)_{L^2(\Sigma)}(\psi_3,\phi_j)_{L^2(\Sigma)}||\psi_1(s)\varphi(t)|$, then

\begin{equation}
|h_\epsilon(s,t,j)|\le g(s,t,j)
\end{equation}

Moreover,

\begin{equation}
\sum_{j}(\psi_2,\phi_j)_{L^2(\Sigma)}(\psi_3,\phi_j)_{L^2(\Sigma)}=\int_\Sigma \psi_2(w)\psi_3(w)\sqrt{h(w)}dw.
\end{equation}

{This implies  the sequence is unconditionally convergent and therefore absolutely convergent.} 

Hence, $g(s,t,j)\in L^1(dt\times ds\times \mu)$, where $\mu$ is the counting measure on $\N$.

Using dominated convergence we obtain 

\begin{align}
&\lim_{\epsilon\rightarrow0^+ }\ang{\ang{F(t+i\epsilon),\psi_1(s)\psi_2(x)\psi_3(y)},\varphi(t)}\\
=&\lim_{n\rightarrow\infty }\int_{\R}\int_{M}\sum_{j}e^{i(t+i\epsilon_n)\lambda_j}e^{is\lambda_j}\left(\int_{\Sigma}\psi_2(x)\phi_j(x)dx\right)\phi_{j}(y)\psi_1(s)\psi_3(y)dyds\varphi(t)dt\\
=&\int_{\R}\int_{M}\sum_{j}\lim_{n\rightarrow\infty }e^{i(t+i\epsilon_n)\lambda_j}e^{is\lambda_j}\left(\int_{\Sigma}\psi_2(x)\phi_j(x)dx\right)\phi_{j}(y)\psi_1(s)\psi_3(y)dyds\varphi(t)dt\\
=&\int_{\R}\int_{M}\sum_{j}e^{it\lambda_j}e^{is\lambda_j}\left(\int_{\Sigma}\psi_2(x)\phi_j(x)dx\right)\phi_{j}(y)\psi_1(s)\psi_3(y)dyds\varphi(t)dt\\
=&\omega^{(2)}_{G}(\varphi\psi_1\psi_2\psi_3).
\end{align}

Therefore $\lim_{\epsilon\rightarrow0^+ }\ang{F(t+i\epsilon),\cdot}=\omega^{(2)}_{G}\in {\mathcal D}'(M\times M)$. Applying \cite[Proposition 7.5]{gerard}

we obtain 

\begin{equation}
WF^{s}(\omega^{(2)}_{G})\subset WF(\omega^{(2)}_{G})\subset \{(\tx,\txi,\ty,\teta)\in T^*(M\times M); \xi_0> 0\}
\end{equation}

which gives $\displaystyle\sum_{\mu=0}^3 g^{0\mu}\xi_\mu=\xi^0>0$.

\end{proof}

\begin{Lemma}\label{causal}
Let $(\tx,\ty)\in M\times M$ be such that $\tilde x$ and $\tilde y$ are not causally related, i.e. $\tx\notin J(\ty)$. Then $(\tx,\txi,\ty,\teta)\notin WF^{-\frac{1}{2}-\epsilon+\tau} (\omega^{(2)}_{G})$.
\end{Lemma}

\begin{proof}
From Eq. \eqref{time}, Lemma \ref{char} and Theorem \ref{positive} we conclude  that 

\begin{equation}
WF^{-\frac{1}{2}-\tep+\tau} (\omega^{(2)}_{G})\subset N_+\times N_{-}
\end{equation}
where $N_{\pm}:=\{(t,x,\xi_0,\xi)\in {\Char}(P); \pm\xi_0>0\}$

Now consider the restriction $\omega^{(2)}_{G}|_{\cal{Q}}:=\omega^{(2)}_{G}:{\D{M\times M}}|_{\cal{Q}}\rightarrow \mathbb{C}$, where the set  $\cal{Q}$  is defined as the set of pairs of causally separated points $(\tx,\ty)\in M\times M$.

Notice that $\omega^{(2)}_{G}= \omega^{+}+i K_G$ where $\omega^+$ is the Schwartz kernel of $A^{-\frac12}\cos(A^{\frac12}(t-s))$ and $K_G$ is the causal propagator. Since, $K_G|_{\cal{Q}}=0$ by \cite[Lemma 5.1]{causal} we have $\omega^{(2)}_{G}|_{\cal{Q}}=\omega^{+}|_{\cal{Q}}$.

Also, the ``flip" map $\rho(\tx,\ty)=(\ty,\tx)$ is a diffeomorphism of $\cal{Q}$ and we have $\rho^*\omega^{+}=\omega^{+}$. Moreover, using  the covariance of the Sobolev wavefront set under diffeomorphisms (see Appendix \ref{diffeo}), we have 

\begin{equation}
WF^{-\frac12-\epsilon+\tau}(\omega^{+}|_{\cal{Q}})=WF^{-\frac12-\epsilon+\tau}(\rho^*\omega^{+}|_{\cal{Q}})=\rho^*WF^{-\frac12-\epsilon+\tau}(\omega^{+}|_{\cal{Q}})
\end{equation}

Moreover, $\rho^*( N_+\times N_{-})= N_-\times N_{+}$ which implies

\begin{equation}
WF^{-\frac12-\epsilon+\tau}(\omega^{(2)}_{G}|_{\cal{Q}})\subset( N_+\times N_{-})\cap( N_-\times N_{+})=\emptyset.
\end{equation}

\end{proof}

\begin{Lemma}\label{diag2}
 If $(\tx,\txi,\tx,\teta)\in WF^{-\frac{3}{2}-\tilde{\epsilon}+\tau}(\omega^{(2)}_{G})$ for some $\tilde{\epsilon}>0$, then $\teta=-\txi$.
\end{Lemma}

\begin{proof}
%Note that $(\tx,\tx)\in\text{sing}\supp(\partial_t \omega^{(2)}_{G})$ and therefore  $(\tx,\tx)\in\text{sing}\supp (\omega^{(2)}_{G})$.
%Assume $(\tilde x, \txi, \tx,\tilde \eta)\in WF^{-\frac32 +\tau-\tilde{\epsilon}}(\omega^{(2)}_{G})$ \blue{and choose $0\le \delta<1$ as $\tau\delta-\epsilon=\tau-\tilde{\epsilon}$ for some $\epsilon>0$. Then $(\tilde x, \txi, \tx,\tilde \eta)\in WF^{-\frac32 -{\epsilon}+\tau\delta}(\omega^{(2)}_{G}).$}
Suppose  $\tilde \eta$ and $\tilde \xi$ are linearly independent, i.e., $\teta\neq \lambda\txi$ for $\lambda\in \mathbb{R}$. By Lemma \ref{char} $(\tilde x, \tilde x, \tilde \xi,\tilde \eta)\in \Char(P)\times\Char(P)$.
% i.e. $\teta$ is not colinear with $\txi$. Then we can c
Now we choose a Cauchy hypersurface $\Sigma_{ t_0}=\{t_0\}\times \Sigma$ such that the null geodesic with initial data $(\tx,\txi)$ and the null geodesic with initial data $(\tx,\teta)$ intersect it. These points of intersections are unique by global hyperbolicity  (see\cite{vien, clemens,minguzzi} for low regularity definitions) . Moreover, using the condition $\teta\neq\lambda\txi$, we can choose $\Sigma_ {t_0}$ such that these points are distinct. We denote these points  by $(t_0,x_0), (t_0, y_0)$. Clearly, these points are not causally related. 

Notice that $\omega_A$ satisfies $P_{(t,x)}\omega_A=0$, 
{$P^b_{(t,x)}\omega_A\in H^{-\frac{3}{2}-\epsilon+\tau}$}. This allows us to choose {$s=-\frac{3}{2}-\epsilon+\tau$}, which is in the range  $0<s<\tau-1$ in Theorem \ref{propagation}.
This propagation of singularities result applied to the distribution $\omega_A$ and the operator $P_{(t,x)}$  guarantees that if  
{$(x,x,\txi,\teta)\in WF^{-\frac{3}{2}-\epsilon+\tau}(\omega^{(2)}_{G})\subset WF^{-\frac{1}{2}-\epsilon+\tau}(\omega_A)$} then the full null bicharacteristic is contained in the wavefront set i.e. $(\gamma(\tx,\txi),(\tx,\teta))\in {WF^{-\frac{1}{2}-\epsilon+\tau}(\omega_A)}$, where $\gamma(\tx,\txi)$ is the null bicharacteristic with initial data $(\tx,\txi)$. Similarly, using the operator $P_{(s,y)}$, we obtain
 $(\gamma(\tx,\txi),\gamma (\tx,\teta))\in {WF^{-\frac{1}{2}-\epsilon+\tau}(\omega_A)}$   where $\gamma(\tx,\teta)$ is the null bicharacteristic with initial data $(\tx,\teta)$.

Now we show that {$(\gamma(\tx,\txi),\gamma (\tx,\teta))\in {WF^{-\frac{1}{2}-\epsilon+\tau}(\omega^{(2)}_{G})}$.}

By Theorem 6.1.1' from \cite{duistermaat} we have 
\begin{equation}\label{propa}
{ WF^{-\frac{1}{2}-\epsilon+\tau}(\omega_A)\backslash 
WF^{-\frac{1}{2}-\epsilon+\tau}(\omega^{(2)}_{G})
= WF^{-\frac{1}{2}-\epsilon+\tau}(\omega_A)\backslash 
WF^{-\frac{1}{2}-\epsilon+\tau}(i\partial_{t}\omega_A)}
\subset\Char (i\partial_{t}).
\end{equation}

However, using Eq.\eqref{char2} and that $(\partial_t+\partial_s)\omega_A=0$, we have the inclusion 
\begin{align}\label{empty}
{WF^{-\frac{1}{2}-\epsilon+\tau}(\omega_A)\subset WF^{\frac{1}{2}-\epsilon+\tau}(\omega_A)}%\nonumber
\subset (\Char(P)\times\Char(P))\cap\Char(\partial_t+\partial_s).
\end{align}

Since $(\Char(P)\times \Char(P))\cap \Char(\partial_t+\partial_s)\cap \Char(i\partial_{t})=\emptyset$, 
taking the intersection between Eq.\eqref{propa} and Eq.\eqref{empty}, we obtain that the left hand side of Eq.\eqref{propa} must be empty. Therefore,
\begin{equation}\label{inclusion}
 {WF^{-\frac{1}{2}-\epsilon+\tau}(\omega_A)\subset WF^{-\frac{1}{2}-\epsilon+\tau}(\omega^{(2)}_{G}).}
\end{equation}
Hence,  $(\gamma(\tx,\txi),\gamma (\tx,\teta))\in {WF^{-\frac{1}{2}-\epsilon+\tau}}(\omega^{(2)}_{G})$. In particular  
$(t_0,x_0,\txi, t_0, y_0, \teta)\in{WF^{-\frac{1}{2}-\epsilon+\tau}}(\omega^{(2)}_{G}))$. However,  this is a contradiction to {Lemma \ref{causal}}. Therefore, $\teta=\lambda\txi$ for some $\lambda\in \R$. Using Eq.(\ref{time}) we have $\xi_0=-\eta_0$ which gives $\lambda=-1$ i.e. $\teta=-\txi$.
\end{proof}

{We have used the distribution $\omega_{A}$, because a direct application of Theorem \ref{elliptic} for $\omega^{(2)}_{G}$ is not possible, since for $\delta$ close to $1$, $\sigma$ cannot take the value  $-\frac{1}{2}$.}

Now we state the main result

\begin{Theorem}\label{main}
{$WF'^{-\frac{3}{2}-{\epsilon}+\tau}(\omega^{(2)}_{G})\subset C^+$ for every ${\epsilon}>0$} and $C^+$ as in Eq.\eqref{C}. % when is wavefrontset empty?
\end{Theorem}

\begin{proof}
Let $(\txx,\txixi)=(\tx,\txi,\ty,-\teta)\in WF^{-\frac{3}{2}-{\epsilon}+\tau}(\omega^{(2)}_{G})
%. Choose $0\le\delta<1$ so that $\tau\delta-\epsilon=\tau-\tilde{\epsilon}$ for some $\epsilon >0$}then $(\txx,\txixi)\in WF^{-\frac{3}{2}-\epsilon+\tau\delta}(\omega^{(2)}_{G})
{\subset WF^{-\frac{1}{2}-\epsilon+\tau}(\omega_{A})} $, where the inclusion follows from  \cite[Proposition B.3]{adiabatic} since $\omega^{(2)}_{G}=i\partial_t\omega_{A}$. 
The propagation of singularities result (Theorem \ref{propagation})  implies that  $(\gamma(\tx,\txi),\gamma (\ty,-\teta))\in { WF^{-\frac{1}{2}-\epsilon+\tau}}(\omega_{A})$
 where  $\gamma(\tx,\txi)$ is the null bicharacteristic with initial data $(\tx,\txi)$ and $\gamma(\ty,-\teta)$ is the null bicharacteristic with initial data $(\ty,-\teta)$. Hence, by Eq.\eqref{inclusion}, we have  $(\gamma(\tx,\txi),\gamma (\ty,-\teta))\in {W F^{-\frac{1}{2}-\epsilon+\tau}}(\omega^{(2)}_{G})$.
Now we choose a Cauchy surface $\Sigma_{t_1}=\{t_1\}\times \Sigma$ and notice that $(t_1,x_1,{\xi}_1,t_1,x_1,{\teta}_1)= (\gamma(\tx,\txi),\gamma(\ty,\teta))\cap (\Sigma_{t_1}^2)$ must satisfy Lemma \ref{diag2} and therefore is of the form $(t_1,x_1,{\txi}_1,t_1,x_1,-{\txi}_1).$ 

{
This observation allows us to define the following curve $\tilde{\gamma}:(-\infty,\infty)\rightarrow M$ 
\begin{equation}
\tilde{\gamma}(t)=\begin{cases}
\Pi\gamma(\tx,\txi)(t)& t=(-\infty, t_1)\\
\Pi\gamma(\ty,-\teta)(-t)& t=(-t_1,-\infty)
\end{cases}
\end{equation}
where  $\Pi$ is the projection from $T^*(M\times M)$ to $M\times M$ and  we assume that $a<t_1<b$, where $\tilde{\gamma}(a)=\Pi\gamma(\tx,\txi)(a)=\tx$ and $\tilde{\gamma}(b)=\Pi\gamma(\ty,-\teta)(b)=\ty$. Moreover $g(\cdot, {\dot{\gamma}})|_{T_{\tx}M}=\txi, g(\cdot, {\dot{\gamma}})|_{T_{\ty}M}=\teta$ and therefore, $\tilde{\gamma}$ is a null geodesic between $\tx$ and $\ty$  with cotangent vectors $\txi$ at $\tx$ and $\teta$ at $\tx$ i.e. $(\txx,\txixi)\in C':=\{(\tx, \txi,\ty, -\tilde{\eta}) ;(\tx, \xi; \ty, \tilde{\eta}) \in C\}$.}

This shows 
\begin{equation}
{WF^{-\frac{3}{2}-\epsilon+\tau}}(\omega^{(2)}_{G})\subset C'
\end{equation}

Using  the definition of $WF^{l '}(u):= \{(\tx, \teta; \ty, -\tilde{\eta}) \in T^{*}(M\times M); (\tx, \txi; \ty, \tilde{\eta}) \in WF^l(u)\}$ and Theorem \ref{positive} gives the result.

\end{proof}

\begin{rem}
For a $C^{1,1}$ metric the same arguments  as used in \cite[Theorem 7.1]{causal} apply and therefore in that scenario we have for every $\epsilon>0$
$$WF^{\frac{1}{2}-\epsilon}(\omega^{(2)}_{G})\subset C'^+.$$
\end{rem}

%One also has the following product theorems

%\begin{Theorem}
%Let $p\in C^\tau S^{m_1}_{1,\delta_1}$ and $q\in S^{m_2}_{1,\delta_2}$, $m_1,m_2\in \mathbb{R}$, $\tau>0$ and $0\le \delta_1,\delta_2<1$. We denote by $p\#q$ the symbol of $p(x,D)q(x,D)$ which has an asymptotic expansion $$p\#q\sim \displaystyle \sum_{\alpha\in \mathbb{N}^n_0}\frac{1}{\alpha!}D^\alpha_\xi p\partial_x^\alpha q.$$
%\end{Theorem}

\section{Appendix}
\subsection{Covariance of the Sobolev Wavefront Set under Diffeomorphisms}\label{diffeo}

\begin{Lemma}
Let $\varphi:M\rightarrow M$ be a $C^\infty$ diffeomorphism and $u\in \mathcal{D}'(M)$. Then 
\begin{equation}
WF^s(\varphi^* u)=\varphi^*WF^s(u),\quad s\in \R.
\end{equation}

\begin{proof}
Let $(x,\xi)\notin WF^s(u)$ which by definition implies { $(\varphi(x),^t\partial\varphi(x)^{-1}\xi)\notin \varphi^*WF^s(u)$.} Moreover, we can write $u=u_1+u_2$ where $u_1\in H^s_{loc}$ and $(x,\xi)\notin WF(u_2)$. By the covariance of the Sobolev spaces in compact sets \cite[Chapter 4, Section 2]{basictheory} we have $\varphi^*u_1\in H^s_{loc}$ and by the covariance under diffeomorphism of the wavefront set $(\varphi(x),^t\partial\varphi(x)^{-1}\xi)\notin WF(\varphi^*u_2)=\varphi^*WF(u_2)$. Putting this together gives $(\varphi(x),^t\partial\varphi(x)^{-1}\xi)\notin WF^s(\varphi^*u)$, i.e. $WF^s(\varphi^*u)\subset \varphi^*WF^s(u)$.

Conversely, let $(y,\eta)\notin WF^s(\varphi^*u)$. Similarly, we obtain $WF^s(\varphi^{-1*}\varphi^*u)\subset (\varphi^{-1*})WF^s(\varphi^*u)$ i.e. $(\varphi^{-1}(y),^t(\partial\varphi^{-1}(y))^{-1}\eta)\notin WF^s(\varphi^{-1*}\varphi^*u)=WF^s(u)$ which implies  $(y,\eta)\notin \varphi^*WF^s(u)$.

% Then, there is $\tilde{u}_1\in H^s_{loc}$ and $\tilde{u}_2$ such that $(y,\eta)\notin WF(\tilde{u}_2)$ that satisfies $\varphi^u=\tilde{u}_1+\tilde{u}_2$.
\end{proof}
\end{Lemma}

%-------------------------------------------------------------------------------------------------------------------------------------------------------
% ACKNOWLEDGEMENTS SECTION<
%-------------------------------------------------------------------------------------------------------------------------------------------------------

{\bf{Acknowledgement.}}
We are grateful to  Chris Fewster, Bernard Kay and James Vickers for helpful discussions.  \\

%-------------------------------------------------------------------------------------------------------------------------------------------------------
% BIBLIOGRAPHY SECTION
%-------------------------------------------------------------------------------------------------------------------------------------------------------
\bibliographystyle{plain}
\bibliography{refs3}

%%%%%%%%%%%%%%%%%%%%%%%%%%%%%%% old biliography %%%%%%%%%%%%%%%%%%%%%%%

\Addresses

%--------------------------------------------------------------- END DOCUMENT --%------------------------------------------------------------
\end{document}